%% file: paper.tex
\documentclass[11pt]{article}
\usepackage{dsfont,amsthm,amsmath,marvosym}
\usepackage{mathptm}
\usepackage{color}
\usepackage[margin=1in]{geometry}
\usepackage{graphicx,enumitem}
\usepackage{cite}
\usepackage{hyperref}

\usepackage{authblk} % allow multiple \author entries

\newtheorem{theorem}{Theorem}

\def\rank{\operatorname{rank}}

\newcommand{\Oo}{\mathcal O} %big-O notation
\newcommand{\Patrascu}{P{\v a}tra{\c s}cu} %Patrascu
\usepackage{amsthm}
\newtheorem{lemma}[theorem]{Lemma}

\newtheorem{corollary}[theorem]{Corollary}

\title{Faster Fully-Dynamic Minimum Spanning Forest}
\author{Jacob Holm\footnote{\texttt{jh@poplar.dk}}}
\author{Eva Rotenberg\footnote{\texttt{roden@diku.dk}}}
\author{Christian Wulff-Nilsen\footnote{\texttt{koolooz@di.ku.dk},
                  \texttt{http://www.diku.dk/$_{\widetilde{~}}$koolooz/}.}}
\affil{Department of Computer Science, University of Copenhagen}

\date{}
\begin{document}
\setcounter{page}{0}

\maketitle
\begin{abstract}
We give a new data structure for the fully-dynamic minimum spanning forest problem in simple graphs. Edge updates are supported in $\Oo(\log^4n/\log\log n)$ amortized time per operation, improving the $\Oo(\log^4n)$ amortized bound of Holm et al. (STOC~'$98$, JACM~'$01$). We assume the Word-RAM model with standard instructions.
\end{abstract}
%\thispagestyle{empty}
%\clearpage

%\section{Introduction}
\input{introduction}

%\input{prelim}

%\section{Reduction to Decremental MSF}\label{sec:Reduction}
\input{reduction}

\input{connectivity}

\input{fastmst}

\section*{Acknowledgments}
We thank Mikkel Thorup for fruitful discussions.

\bibliographystyle{amsplain}
\bibliography{general}

\end{document}

%% file: introduction.tex
\section{Introduction}\label{sec:Intro}
%Why MSF is a natural problem
A dynamic graph problem is that of maintaining a dynamic graph on $n$ vertices where edges may be inserted or deleted and possibly where queries regarding properties of the graph are supported. We call the dynamic problem decremental resp.~incremental if edges can only be deleted resp.~inserted, and fully dynamic if both edge insertions and deletions are supported.
%On-line?
%Which operations we support

We consider the fully-dynamic minimum spanning forest (MSF) problem which is to maintain a state for each edge of whether it belongs to the current MSF or not. After an edge update, at most one edge becomes a new tree edge in the MSF and at most one edge becomes a non-tree edge and a data structure needs to output which edge changes state, if any.
%\begin{itemize}
%\item supports the operation insert(u,v), inserting the edge $(u,v)$,
%\item supports the operation delete(u,v), deleting the edge $(u,v)$,
%\item and maintains a top-tree forest of the MSF.
%%treepath(u,v), returning the tree-path from $u$ to $v$
%\end{itemize}

%Running times

Dynamic MSF was first studied by Frederickson~\cite{frederickson85connect} who achieved a worst-case update time of $\Oo(\sqrt{m})$ where $m$ is the number of edges at the time of the update. This was later improved by Eppstein et al.~\cite{eppstein97sparsification} to $\Oo(\sqrt{n})$ using the sparsification technique. Henzinger and King made a data structure with amortized update time $\Oo(\sqrt[3]{n}\log n)$. Holm et al.~\cite{holm01connect} dramatically improved this amortized bound to $\Oo(\log^4 n)$.  All these bounds are for simple graphs (no parallel edges), but any MSF structure can be extended to general graphs via a simple reduction that adds $\Oo(\log m)$ to the update time.  In the following we will assume all graphs are simple unless otherwise stated.

We show how to support updates in $\Oo(\log^4n/\log\log n)$ amortized time, improving the bound by Holm et al. To obtain this bound, we assume the RAM model of computation with standard instructions.
More generally, our time bound per update can be written as
\[
  \Oo\left(\frac{\log^4 n}{\log\log n}\cdot \frac{\mathit{sort}(\log^{c}n,n^2)}{\log^{c}n}\right),
\]
for some constant $c > 0$, where $\mathit{sort}(k,r)$ is the time for sorting $k$ natural numbers with values in the range from $0$ to $r$. Equivalenty, $\mathit{sort}(k,r)/k$ is the operation time of a priority queue. Thus, the update time of our structure depends on the model of computation, and the choice of the priority queue that our structure uses as a building block. The following table shows both deterministic and randomized variants of the data structure differing only in the choice of priority queue.\\
\begin{center}
\begin{table}[ht!]
\begin{tabular}{|l|l|l|}
\hline
 & Deterministic & Randomized \\
\hline
RAM w. AC$^0$ & $\Oo(\log ^4 n \sqrt{\log \log \log n}/\sqrt{\log \log n})$ & $\Oo(\log^4 n \log \log \log n/\log \log n)$ \\
\hline
RAM, AC$^0$, $\Oo(1)$ multiplication & $\Oo(\log^4 n \log \log \log n/\log \log n)$ & $\Oo(\log ^4 n/\log \log n)$ \\
\hline
\end{tabular}
\caption{Our update time, depending on the choice of priority queue from \cite{Raman96fastalgorithms,thorup02sortAC0,han04detsort,andersson98sorting}, see Section \ref{relwork}}
\end{table}
\end{center}

%Previous work
\subsection{Related Work}\label{relwork}
Holm et al.~\cite{holm01connect} gave a deterministic data structure for decremental MSF with $\Oo(\log^2 n)$ amortized update time. Combining this with a slightly modified version of a reduction from fully-dynamic to decremental MSF of Henzinger and King~\cite{Henzinger97fullydynamic}, they obtained their $\Oo(\log^4n)$ bound for fully-dynamic MSF. A somewhat related problem to dynamic MSF is fully-dynamic connectivity. Here
a data structure needs to support insertion and deletion of edges as well as connectivity queries between vertex pairs. The problem was first studied by Frederickson~\cite{frederickson85connect} who obtained $\Oo(\sqrt{m})$ update time $\Oo(1)$ query time data structure. Update time was improved to $\Oo(\sqrt n)$ by Eppstein et al.~\cite{eppstein97sparsification}. Henzinger and King~\cite{henzinger99connect} obtained expected $\Oo(\log^3n)$ amortized update time and query time $\Oo(\log n / \log \log n)$. Henzinger and Thorup~\cite{henzinger97sampling} improved update time to $\Oo(\log^2n)$ with a clever sampling technique. A deterministic structure with the same bounds was given by Holm et al.~\cite{holm01connect}. Thorup~\cite{thorup00connect} achieved an expected amortized update-time of $\Oo(\log n (\log\log n)^3)$ and query time $\Oo(\log n/\log\log\log n)$, using randomization. Wulff-Nilsen~\cite{Wulff-Nilsen13a} gave a deterministic, amortized $\Oo(\log ^2 n/\log\log n)$ update-time data structure with $\Oo(\log n /\log\log n)$ query time. An $\Omega(\log n)$ lower bound on the operation time for fully-dynamic connectivity and MSF was given by \Patrascu\ and Demaine~\cite{patrascu04connect}.% They show that if a dynamic connectivity or MSF algorithm has update time $t_u$ and query time $t_q$ then $\max(t_u,t_q)\in \Omega(\log n)$.
 
%He also proves trade-offs which show that the algorithms of Thorup and Wulff-Nilsen are ...

%Priority queues Raman96fastalgorithms,thorup02sortAC0,han04detsort,andersson98sorting
As indicated above, priority queues are essential to our data structure. Equivalently, we rely on the ability to efficiently sort $l= \log^cn$ elements from $[n^2]$ where $c$ is a constant. Expressed as a function of $l$, the elements lie in the range $0\ldots 2^{w}-1$, where $w = 2 l^{1/c}$. To sort quickly, we rely on $w>l$.
In the RAM-model with AC$^0$ instructions, Raman~\cite{Raman96fastalgorithms} gave a deterministic bound of $\Oo(l\sqrt{\log l \log \log l})$. Using randomization, Thorup~\cite{thorup02sortAC0} improved this to $\Oo(l\log \log l)$. The same time bounds were achieved without randomization, if assuming constant time multiplication, by Han~\cite{han04detsort}. Andersson et al.~\cite{andersson98sorting} achieve optimal $\Oo(l)$ sorting time, using randomization, and assuming $\Oo(1)$ time multiplication; their algorithm requires $w\gg \log^{2+\varepsilon} l$ for some constant $\varepsilon$, which in our case is satisfied as $w > l^{1/c}$.

%Most recently, Holm et al describe first a decremental MSF structure with update time $\Oo(\log ^2 n)$ based on the same ideas as their fully dynamic connectivity data structure with update time $\Oo(\log ^2 n)$, and then a reduction to fully dynamic MSF with resulting update time $\Oo(\log ^4 n)$.

\subsection{Idea and paper outline}
Since the data structures of Holm et al.~\cite{holm01connect} for decremental MSF and fully dynamic connectivity are essentially the same, the question arises of whether the $\Oo(\log ^2 n/\log \log n)$ fully-dynamic connectivity structure in~\cite{Wulff-Nilsen13a} can be directly translated to an improved $\Oo(\log ^2 n /\log \log n)$ decremental MSF structure. If that were the case, we could immediately use the reduction from fully-dynamic to decremental MSF in~\cite{holm01connect} to obtain an $\Oo(\log ^4 n / \log \log n)$ bound for fully-dynamic MSF. Unfortunately, that is not the case as the data structure in~\cite{Wulff-Nilsen13a} relies on a shortcutting system which can not be easily adapted to decremental MSF.
%The algorithm for finding a replacement edge relies on using short-cuts to visit some edges, the level $i$ edges, ordered according to the cluster forest - not ordered by weight. 
Instead, we make a different analysis of the reduction from decremental to fully dynamic MSF (Section~\ref{sec:Reduction}) which surprisingly shows how a slightly slower decremental MSF structure than that in~\cite{holm01connect} can in fact lead to a slightly faster fully dynamic MSF!

A modified version of the dynamic connectivity structure by Wulff-Nilsen ~\cite{Wulff-Nilsen13a} with $\Oo(\log^2n)$ update time is described in Section~\ref{sec:Connectivity}. It is shown in Section~\ref{subsec:DecMST} how to modify it to a simple decremental MSF structure with the same performance. We then show how to speed up a certain part of this decremental MSF structure in Section~\ref{sec:FastMSF}. The main idea is to extend it with a non-trivial shortcutting system involving fast priority queues in order to speed up the search for replacement edges. This system is the main technical contribution of the paper. We conclude Section~\ref{sec:FastMSF} by showing that this data structure for decremental MSF speeds up fully-dynamic MSF.

%and The $\log \log n$ speed-up relies on the usage of short-cuts in a tree over possible replacement-edges. In order to still use short-cuts, for decremental MSF, we need to maintain a priority queue of the cheapest level $i$ edge below a node in the cluster forest. Maintaining priority queues does in fact lead to a slower decremental MSF, but as desired, it distributes time-consumption wisely, leading to an over-all speed-up in the data structure. The complete analysis of the fast MSF is found in Section~\ref{sec:FastMSF}.

%% file: reduction.tex
\section{Reduction to decremental MSF}\label{sec:Reduction}
%The reduction from fully-dynamic to decremental MSF described by Holm et al. (which is based on King).
In this section, we present a different analysis of the reduction from decremental MSF to fully dynamic MSF from~\cite{holm01connect} based on the construction from~\cite{Henzinger97fullydynamic}. The main difference is that in our analysis, we do not insist on all edges being deleted in the decremental MSF problem.
%$n$ and $m$ respectively denote the number of vertices and edges of the fully dynamic graph.
\begin{lemma}\label{Lem:Reduction}
Suppose we have a decremental (deletions-only) MSF data structure that for a connected simple graph with $n$ vertices and $m$ edges has a total worst-case running time for the construction and the first $d$ deletions of $\Oo(t_cm + t_rd)$, where $t_c$ and $t_r$ are non-decreasing functions of $n$.
%where each edge is initialized with $c$ credits, each credit representing a constant amount of time, and where the time per delete-edge operation can be split into two summands:
%\begin{itemize}
%\item an amortized cost paid by credits associated with edges,
%\item an additional cost of at most $t_r$.
%in which an edge is initialised with $c$ credits (for amortisation), each representing a constant amount of time, and where, given a deletion, a replacement edge can be found in worst-case time at most $t_r$, 
%\end{itemize}
Then there exists a fully dynamic MSF data structure for simple graphs on $n$ vertices with amortized update time $\Oo( \log^3 n + t_c \log ^2 n + t_r \log n )$.
\end{lemma}
\begin{proof}
Let $G$ be the fully dynamic simple graph with $n$ vertices and up to $m=\Oo(n^2)$ edges. We now describe the fully dynamic data structure to maintain the MSF $F$ of $G$.
%Let $G$ be the fully dynamic graph, and let $n,m$ denote the fixed number of vertices and varying number of edges, respectively.

Keep track of an array of at most $\lceil\lg m\rceil\leq\lceil2\lg n\rceil$ decremental graphs $A_i$ with non-tree edge count $|A_i| \le 2^i$ (we call this the \emph{edge-count invariant}).  Each $A_i$ corresponds to a (not necessarily connected) subgraph of $G$, where tree paths in $G$ may be represented by single edges.  We use the decremental MSF data structure for each component of each $A_i$, and maintain the invariant that each non-tree edge of $G$ is a non-tree edge of some $A_i$ (we call this the \emph{non-tree edge invariant}).  

When an edge $e$ is inserted, we use a top-tree over $F$ to determine whether $e$ becomes a tree-edge, possibly replacing some edge $e^\prime$ in which case $e^\prime$ is identified by the top-tree. The insert operation may create a new non-tree edge $e^{\prime \prime}$, in which case we must initialize $e^{\prime\prime}$ in a new decremental structure, in order to maintain the non-tree edge invariant.
%, we may have to initialize the new non-tree ($e$ or $e^\prime$ or none) as a non-tree edge in some decremental structure. %Call the operation batch-insert($\{e\}$) or batch-insert($\{e^\prime\}$).
%When an edge is inserted, then 
%itself the tree-edge it replaces 
%it must belong to a decremental structure. 
%Batch-insert of $D$.
%We now describe a more general initialize-operation of the (possibly 
%singleton) set of edges, $D$, needed later.
To make sure the edge-count invariant is maintained, we may have to \emph{collapse} decremental structures $A_1, \ldots , A_j$ to create a new $A_j$, for some $j$. 
In general, let $D$ be a (possibly singleton) set of inserted edges. %That is, $D=\{e\}$ or $D=\{e^\prime\}$. 
Choose minimally $j$ such that we can construct $A_j$  from $D \cup F \cup \bigcup_{i\le j} A_i$ without breaking the edge-count invariant. When constructing $A_j$ from the set, we keep all non-tree edges, but may introduce super-edges instead of tree paths (see~\cite{holm01connect} for details). The total number of edges in the resulting $A_j$ is at most $5$ times the number of non-tree edges, and each component of $A_j$ contains at least one non-tree edge.  The time to find the vertices and edges to put into $A_j$ is $\Oo(\log n)$ per non-tree edge.  Since $j$ was chosen minimally, each collapse ensures that at least $2^{j-1}$ non-tree edges come from $D\cup\bigcup_{i<j}A_i$. % a smaller $A_i$. 
%After the insertion of an edge, its decremental structure may be collapsed at most $2\lceil \lg (n^2)\rceil = \Oo(\log n)$ times, by minimality of $j$. 
So if we associate $10(t_c+\log n)(\left\lceil2\lg n\right\rceil-i)$ credits with each non-tree edge in $A_i$, the construction of $A_j$ can be paid for by the non-tree edges that came from $D\cup\bigcup_{i<j}A_i$. %a smaller $A_i$.  
Thus, when an edge is inserted, it must be given $\Oo(\log^2 n+t_c\log n)$ credits for the amortisation. 

Upon an edge deletion, delete($e$), ask each decremental structure containing $e$ for a replacement edge. It follows from the non-tree edge invariant that the cheapest of the edges returned is the desired replacement edge. Since there are $\Oo(\log n)$ decremental structures, each with $\leq n$ vertices, the time for this operation is $\Oo(\log^2 n + t_r \log n)$ plus the time paid for by credits on the edges.
However, the up to $\lceil2\lg n \rceil$ returned replacement-candidate edges have now become tree-edges in their respective decremental structures, possibly violating the non-tree edge invariant, as only one of them joins $F$. To make sure the non-tree edge invariant holds, a decremental structure with the returned edges is created. That is, these edges play the role of $D$ in the insert description above. Each reinitialized edge must be given $\Oo(\log^2 n + t_c \log n)$ credits, and there were $\Oo(\log n)$ replacement candidates, yielding an amortized deletion time of $\Oo(\log^3 n + t_c\log^2 n+t_r\log n)$.
\end{proof}

The following corollary is crucial in obtaining our improvement for fully-dynamic MSF. It shows that to obtain a faster data structure for this problem by reduction to decremental MSF, it actually suffices with a decremental MSF structure which is slower than that in~\cite{holm01connect} in the case where all edges end up being deleted.
%\begin{proof}
%
%\end{proof}
\begin{corollary}\label{Cor:Reduction}
Given a decremental MSF structure with $t_c = \frac{\log ^2 n}{\epsilon\log\log n}$ and $t_r = \log^{2+\epsilon}n$ where $\epsilon < 1$ is a constant, the reduction gives a fully dynamic MSF structure with amortized update time $\Oo(\frac{\log^4}{\log\log n})$.
\end{corollary}

%% \begin{lemma}
%%   Suppose there is a fully dynamic MSF data structure for simple graphs on $n$ vertices with amortized update time $\Oo(f(n))$.  Then there exists a fully dynamic MSF data structure for general graphs on $n$ vertices with amortized update time $\Oo(\log m + f(n))$ where $m$ is the current number of edges.
%% \end{lemma}
%% \begin{proof}
%%   We can represent each bundle of parallel edges in a general graph with a priority queue and just keep the cheapest edge of the bundle in the MSF structure.  An update to the general graph will then correspond to a constant number of priority queue operations and MSF operations on the simple graph.  This costs only an additional $\Oo(\log m)$ time per update.
%% \end{proof}

%% file: connectivity.tex
\section{Simple Data Structures for Dynamic Connectivity and Decremental MSF}\label{sec:Connectivity}
In this section, we give a description of the fully-dynamic connectivity data structure in~\cite{Wulff-Nilsen13a} (which is based on an earlier structure of Thorup \cite{thorup00connect}) except that shortcuts are omitted and a spanning forest is maintained. We will modify it in Section~\ref{subsec:DecMST} to support decremental MSF.

Let $G = (V,E)$ denote the dynamic graph. The data structure maintains, for each edge $e\in E$, a \emph{level} $\ell(e)$ which is an integer between $0$ and $\ell_{\max} = \lfloor\log n\rfloor$. As we shall see, the level of an edge $e$ starts at $0$ and can only increase over time and for the amortization, we can view $\ell_{\max} - \ell(e)$ as the amount of credits left on $e$.

For $0\leq i\leq\ell_{\max}$, let $E_i$ denote the set of edges of $E$ with level at least $i$ and let $G_i = (V,E_i)$. The (connected) components of $G_i$ are \emph{level $i$ clusters} or just \emph{clusters}. The following invariant is maintained:
\begin{description}
\item [Invariant:] For each $i$, any level $i$ cluster spans at most $\lfloor n/2^i\rfloor$ vertices.
\end{description}

Consider a level $i$ cluster $C$. By contracting all edges of $E_{i+1}$ in $C$, we get a connected multigraph of level $i$-edges where vertices correspond to level $(i+1)$ clusters contained in $C$. Our data structure maintains a spanning tree of this multigraph. The union of spanning trees over all clusters is a spanning forest of $G$.

The data structure maintains a \emph{cluster forest} of $G$ which is a forest $\mathcal C$ of rooted trees where a node $u$ at level $i$ is a level $i$ cluster $C(u)$. Roots of $\mathcal C$ are components of $G = G_0$ and leaves of $\mathcal C$ are vertices of $G$.
A level $i$-node $u$ which is not a leaf has as children the level $(i+1)$ nodes $v$ for which $C(v)\subseteq C(u)$. In addition to $\mathcal C$, the data structure maintains $n(u)$ for each node $u\in\mathcal C$ denoting the number of vertices of $G$ contained in $C(u)$ (equivalently, the number of leaves in the subtree of $\mathcal C$ rooted at $u$).

\subsection{Handling insertions and deletions}\label{subsec:InsDel}
When a new edge $e = (u,v)$ is inserted into $G$, it is given level $0$ and $\mathcal C$ is updated by merging the roots $r_u$ and $r_v$ corresponding to the components of $G$ containing $u$ and $v$, respectively. The new root inherits the children of both $r_u$ and $r_v$. If $r_u\neq r_v$, $e$ becomes a tree edge in the new level $0$ cluster. Otherwise, $e$ becomes a non-tree edge.

Deleting an edge $e = (u,v)$ is more involved. If $e$ is not a tree edge, no structural changes occur in $\mathcal C$. Otherwise, let $i = \ell(e)$. The deletion of $e$ splits a spanning tree of a level $i$ cluster $C$ into two subtrees, $T_u$ containing $u$ (inside some level $i+1$-cluster) and $T_v$ containing $v$. One of these trees, say $T_u$, contains at most half the vertices (in $V$) of $C$. For each level $i$ edge in $T_u$, we increase its level to $i+1$. In $\mathcal C$, this amounts to merging all nodes corresponding to level $i+1$ clusters in $T_u$ into one node, $w$; see Figure~\ref{fig:SplitMerge}(a) and (b). By the choice of $T_u$, this does not violate the invariant.
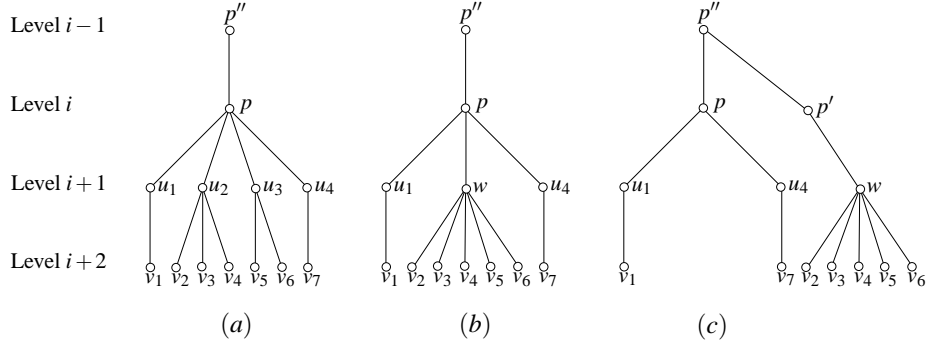
\begin{figure}
\centerline{\scalebox{0.7}{\input{SplitMerge.pstex_t}}}
\caption{(a): Part of $\mathcal C$ before the merge. (b): Level $i+1$ nodes $u_2$ and $u_3$ are merged into a new level $i+1$ node $w$. (c): A replacement level $i$ edge was not found so $w$ is given a new parent $p'$ which becomes the sibling of $p$.}
\label{fig:SplitMerge}
\end{figure}

Next, we search through (non-tree) level $i$ edges incident to $C(w)$ in some arbitrary order until some edge is found which connects $C(w)$ and $T_v$ (if any). For all visited level $i$ edges which did not reconnect the two trees, their level is increased to $i+1$, thereby paying for them being visited. If a replacement edge $(a,b)$ was found, no more structural changes occur in $\mathcal C$ and $(a,b)$ becomes a new tree edge. Otherwise, $w$ is removed from its parent $p$ (corresponding to $C$) and a new level level $i$ node $p'$ is created having $w$ as its single child and having $p$ as sibling; see Figure~\ref{fig:SplitMerge}(b) and (c). This has the effect of splitting $C = C(p)$ into two smaller level $i$ clusters. The same procedure is now repeated recursively at level $i-1$ where we try to reconnect the two trees of level $i-1$ edges containing the new level $i$ clusters $C(p)$ and $C(p')$, respectively. If level $0$ is reached and no replacement edge was found, a component of $G$ is split in two.

\subsection{Local trees}\label{subsec:LocalTrees}
To guide the search for level $i$ tree/non-tree edges, we first modify $\mathcal C$ to a forest $\mathcal C_L$ of binary trees. This is done by inserting, for each non-leaf node $u\in\mathcal C$, a binary \emph{local tree} $L(u)$ between $u$ and its children; see Figure~\ref{fig:LazyLocalTree}.
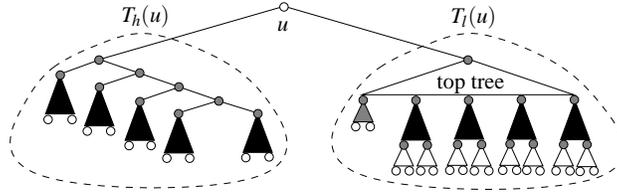
\begin{figure}
\centerline{\scalebox{0.7}{\input{LazyLocalTree.pstex_t}}}
\caption{The structure of local tree $L(u)$ of a node $u$ in $\mathcal C$, from~\cite{Wulff-Nilsen13a}. In $T_h(u)$, rank trees are black and the rank path and roots of rank trees are grey. In $T_l(u)$, the buffer tree is grey, top and bottom trees are white, and rank trees are black.}
\label{fig:LazyLocalTree}
\end{figure}
To describe the structure of $L(u)$, we first need to define heavy and light children of $u$. A child $v$ of $u$ in $\mathcal C$ is \emph{heavy} if $n(v)\geq n(u)/\log^{\epsilon_h}n$, where $\epsilon_h > 0$ is a constant that we may pick as small as we like. Otherwise, $v$ is \emph{light}.

The root of $L(u)$ has two children, one rooted at \emph{heavy tree} $T_h(u)$ and the other rooted at \emph{light tree} $T_l(u)$. The leaves of $T_h(u)$ resp.~$T_l(u)$ are the heavy resp.~light children of $u$. Before describing the structure of these trees, let us associate a \emph{rank} $\rank(v)\leftarrow \lfloor\log n(v)\rfloor$ to each node $v$ in $\mathcal C$.

Tree $T_h(u)$ is formed by initially regarding each heavy child of $u$ as a trivial rooted tree with a single node and repeatedly pairing roots $r_1$ and $r_2$ of trees with the same rank, creating a new tree with a root $r$ of rank $\rank(r) = \rank(r_1) + 1$ and with children $r_1$ and $r_2$. When the process stops, the remaining rooted trees, called \emph{rank trees}, all have distincts ranks and they are attached as children to a rooted \emph{rank path} $P$ such that children with larger rank are closer to the root of $P$ than children of smaller rank. We define the rank of a node on the rank path to be the larger of the ranks of its children.

Tree $T_l(u)$ is more involved. Its leaves are the light children of $u$ and they are divided into groups each having size at most $\log^\alpha n$, where $\alpha$ is a constant that we may pick as large as we like. The nodes in each group are kept as leaves in a balanced binary search tree (BBST) ordered by $n(v)$-values. One of these trees is the \emph{buffer tree} and the others are \emph{bottom trees}. We define the rank of each bottom tree root as the maximum rank of its leaves and we pair them up into rank trees exactly as we did when forming $T_h(u)$. However, instead of attaching the rank tree roots to a rank path, we instead keep them as leaves of a BBST called the \emph{top tree}, where again the ordering is by rank. We also have the buffer tree root as a child of the top tree and we regard it as having smaller rank than all the other leaves.

It was shown in~\cite{Wulff-Nilsen13a} that $\mathcal C_L$ has height $\Oo(\frac 1{\epsilon_h}\log n)$. Refer to nodes of $\mathcal C_L$ belonging to $\mathcal C$ as \emph{cluster nodes}.

\paragraph{Merging local trees}
We need to support the merge of local trees $L(u)$ and $L(v)$ in $\mathcal C_L$ corresponding to a merge of cluster nodes $u$ and $v$ into a new node $w$. First, we merge the buffer trees of $L(u)$ and $L(v)$ into a new BBST $T_b$ by adding the leaves of the smaller tree to the larger tree. Heavy trees $T_h(u)$ and $T_h(v)$ have their rank paths removed and leaves that should be light in $L(w)$ are removed from $T_h(u)$ and $T_h(v)$ and added as leaves of $T_b$. For each leaf removed from $T_h(u)$ and $T_h(v)$, we remove their ancestor rank nodes. We end up with subtrees of the original rank trees in $T_h(u)$ and $T_h(v)$ and these subtrees are paired up as before and attached to a new rank path for $T_h(w)$. Tree $T_b$ becomes a buffer tree in $T_l(w)$ if its number of leaves does not exceed $\log^\alpha n$; otherwise, it becomes a bottom tree in $T_l(w)$, leaving an empty buffer tree. Rank trees in $T_l(u)$ and $T_l(v)$ are stripped off from their top trees and paired up into new rank trees as before (here we include $T_b$ if it became a bottom tree) and these are attached as leaves to a new top tree for $T_l(w)$.

In the above merge, let $p$ be the parent of $u$ and $v$ in $\mathcal C$. In $\mathcal C_L$, we need to delete $u$ and $v$ as leaves of $L(p)$ and to add $w$ as a new leaf of $L(p)$. We shall only describe the deletion of $u$ as $v$ is handled in the same manner. We consider four cases depending on which part of $L(p)$ $u$ belongs to:
\begin{itemize}
\item If $u$ is a leaf in the buffer tree of $T_l(p)$, we delete it with a standard BBST operation in that tree.
\item If $u$ is a leaf in a bottom tree $B$ of $T_l(p)$, a similar BBST update happens in $B$. Additionally we update the max rank of leaves in $B$ as this rank is associated with the root of $B$. If the maximum does not decrease, no further updates are needed. Otherwise, we remove all ancestor rank nodes of $B$ in $T_l(u')$, pair the resulting rank trees as before and attach them as leaves of the top tree.
\item If $u$ is a leaf in $T_h(p)$, we remove it and its ancestor rank nodes in $T_h(p)$, pair up the resulting rank trees and attach them to a new rank path for $T_h(p)$.
\end{itemize}
To add $w$ as a new leaf of $L(p)$, we only have two cases. If $w$ is a heavy node, we regard it as a trivial rank tree, delete the rank path of $T_h(p)$, repeatedly pair up the rank trees (including $w$) and reattach them with a new rank path to form the updated $T_h(p)$. If instead $w$ is a light node, we add it to the buffer tree of $T_l(p)$ (which may be turned into a bottom tree, as described above).

\paragraph{Handling cluster splits}
What remains is to describe the updates to local trees after splitting a level $i$ cluster in two. Let $w$, $p$, and $p'$ be defined as in the previous subsection and let $p''$ be the parent of $p$ and $p'$ in $\mathcal C$ (Figure~\ref{fig:SplitMerge}(b) and (c)). Creating $L(p')$ is trivial as $p'$ has only the single child $w$ in $\mathcal C$ and attaching $p'$ as a leaf of $L(p'')$ is done as above. The removal of $w$ from $L(p)$ decreases $n(p)$ which may cause some light children of $p$ in $\mathcal C$ to become heavy. In $\mathcal C_L$, each corresponding leaf of $T_l(p)$ is removed and added to $T_h(p)$ and $L(p)$ is updated accordingly as described above. Since $n(p)$ decreases, $p$ might change from being a heavy child of $p''$ to being a light child. If so, we move it from $T_h(p'')$ to the buffer tree of $T_l(p'')$, as described above.

\paragraph{Bitmaps}
Having modified $\mathcal C$ into the forest $\mathcal C_L$ of binary trees, we add bitmaps to nodes of $\mathcal C_L$ to guide the search for level $i$ edges. More precisely, each node $u\in\mathcal C_L$ is associated with two bitmaps $\mathit{tree}(u)$ and $\mathit{nontree}(u)$, where the $i$th bit of $\mathit{tree}(u)$ ($\mathit{nontree}(u)$) is $1$ iff there is at least one level $i$ tree (non-tree) edge incident to a leaf in the subtree of $\mathcal C_L$ rooted at $u$. Since $\mathcal C_L$ is binary, these bitmaps enable us to identify a level $i$ tree/non-tree edge incident to a cluster $C(u)$ by traversing a path down from $u$ in $\mathcal C_L$ in time proportional to its length by backtracking when bitmaps with $i$th bit $0$ are encountered. When a level $i$ tree edge is removed (which happens if it is deleted from $G$ or has its level increased), then for each of its endpoints $u$, we set $\mathit{tree}(u)[i] = 0$ and update the bitmaps for all ancestors $v$ of $u$ in $\mathcal C_L$ bottom-up by taking the logical 'or' of the $\mathit{tree}$-bitmaps of its children. A similar update is done to $\mathit{nontree}$-bitmaps if $u$ is a non-tree edge. When inserting a level $i$ tree/non-tree edge, bitmaps are updated in a similar manner.

\input{simplemst}

%% file: SplitMerge.pstex_t
\begin{picture}(0,0)%
\includegraphics{SplitMerge.pstex}%
\end{picture}%
\setlength{\unitlength}{4144sp}%
\begingroup\makeatletter\ifx\SetFigFont\undefined%
\gdef\SetFigFont#1#2#3#4#5{%
  \reset@font\fontsize{#1}{#2pt}%
  \fontfamily{#3}\fontseries{#4}\fontshape{#5}%
  \selectfont}%
\fi\endgroup%
\begin{picture}(7763,2939)(3069,-7683)
\put(9840,-7168){\makebox(0,0)[lb]{\smash{{\SetFigFont{12}{14.4}{\familydefault}{\mddefault}{\updefault}{\color[rgb]{0,0,0}$v_2$}%
}}}}
\put(10065,-7168){\makebox(0,0)[lb]{\smash{{\SetFigFont{12}{14.4}{\familydefault}{\mddefault}{\updefault}{\color[rgb]{0,0,0}$v_3$}%
}}}}
\put(10290,-7168){\makebox(0,0)[lb]{\smash{{\SetFigFont{12}{14.4}{\familydefault}{\mddefault}{\updefault}{\color[rgb]{0,0,0}$v_4$}%
}}}}
\put(10401,-6395){\makebox(0,0)[lb]{\smash{{\SetFigFont{12}{14.4}{\familydefault}{\mddefault}{\updefault}{\color[rgb]{0,0,0}$w$}%
}}}}
\put(8955,-4915){\makebox(0,0)[lb]{\smash{{\SetFigFont{12}{14.4}{\familydefault}{\mddefault}{\updefault}{\color[rgb]{0,0,0}$p''$}%
}}}}
\put(9090,-5680){\makebox(0,0)[lb]{\smash{{\SetFigFont{12}{14.4}{\familydefault}{\mddefault}{\updefault}{\color[rgb]{0,0,0}$p$}%
}}}}
\put(9630,-7165){\makebox(0,0)[lb]{\smash{{\SetFigFont{12}{14.4}{\familydefault}{\mddefault}{\updefault}{\color[rgb]{0,0,0}$v_7$}%
}}}}
\put(8280,-7165){\makebox(0,0)[lb]{\smash{{\SetFigFont{12}{14.4}{\familydefault}{\mddefault}{\updefault}{\color[rgb]{0,0,0}$v_1$}%
}}}}
\put(8392,-6383){\makebox(0,0)[lb]{\smash{{\SetFigFont{12}{14.4}{\familydefault}{\mddefault}{\updefault}{\color[rgb]{0,0,0}$u_1$}%
}}}}
\put(9732,-6378){\makebox(0,0)[lb]{\smash{{\SetFigFont{12}{14.4}{\familydefault}{\mddefault}{\updefault}{\color[rgb]{0,0,0}$u_4$}%
}}}}
\put(10503,-7165){\makebox(0,0)[lb]{\smash{{\SetFigFont{12}{14.4}{\familydefault}{\mddefault}{\updefault}{\color[rgb]{0,0,0}$v_5$}%
}}}}
\put(10747,-7165){\makebox(0,0)[lb]{\smash{{\SetFigFont{12}{14.4}{\familydefault}{\mddefault}{\updefault}{\color[rgb]{0,0,0}$v_6$}%
}}}}
\put(9968,-5718){\makebox(0,0)[lb]{\smash{{\SetFigFont{12}{14.4}{\familydefault}{\mddefault}{\updefault}{\color[rgb]{0,0,0}$p'$}%
}}}}
\put(6922,-4916){\makebox(0,0)[lb]{\smash{{\SetFigFont{12}{14.4}{\familydefault}{\mddefault}{\updefault}{\color[rgb]{0,0,0}$p''$}%
}}}}
\put(7057,-5681){\makebox(0,0)[lb]{\smash{{\SetFigFont{12}{14.4}{\familydefault}{\mddefault}{\updefault}{\color[rgb]{0,0,0}$p$}%
}}}}
\put(7597,-7166){\makebox(0,0)[lb]{\smash{{\SetFigFont{12}{14.4}{\familydefault}{\mddefault}{\updefault}{\color[rgb]{0,0,0}$v_7$}%
}}}}
\put(6247,-7166){\makebox(0,0)[lb]{\smash{{\SetFigFont{12}{14.4}{\familydefault}{\mddefault}{\updefault}{\color[rgb]{0,0,0}$v_1$}%
}}}}
\put(6472,-7166){\makebox(0,0)[lb]{\smash{{\SetFigFont{12}{14.4}{\familydefault}{\mddefault}{\updefault}{\color[rgb]{0,0,0}$v_2$}%
}}}}
\put(6697,-7166){\makebox(0,0)[lb]{\smash{{\SetFigFont{12}{14.4}{\familydefault}{\mddefault}{\updefault}{\color[rgb]{0,0,0}$v_3$}%
}}}}
\put(6922,-7166){\makebox(0,0)[lb]{\smash{{\SetFigFont{12}{14.4}{\familydefault}{\mddefault}{\updefault}{\color[rgb]{0,0,0}$v_4$}%
}}}}
\put(7147,-7166){\makebox(0,0)[lb]{\smash{{\SetFigFont{12}{14.4}{\familydefault}{\mddefault}{\updefault}{\color[rgb]{0,0,0}$v_5$}%
}}}}
\put(7372,-7166){\makebox(0,0)[lb]{\smash{{\SetFigFont{12}{14.4}{\familydefault}{\mddefault}{\updefault}{\color[rgb]{0,0,0}$v_6$}%
}}}}
\put(6359,-6384){\makebox(0,0)[lb]{\smash{{\SetFigFont{12}{14.4}{\familydefault}{\mddefault}{\updefault}{\color[rgb]{0,0,0}$u_1$}%
}}}}
\put(7699,-6379){\makebox(0,0)[lb]{\smash{{\SetFigFont{12}{14.4}{\familydefault}{\mddefault}{\updefault}{\color[rgb]{0,0,0}$u_4$}%
}}}}
\put(7033,-6393){\makebox(0,0)[lb]{\smash{{\SetFigFont{12}{14.4}{\familydefault}{\mddefault}{\updefault}{\color[rgb]{0,0,0}$w$}%
}}}}
\put(4906,-4921){\makebox(0,0)[lb]{\smash{{\SetFigFont{12}{14.4}{\familydefault}{\mddefault}{\updefault}{\color[rgb]{0,0,0}$p''$}%
}}}}
\put(5041,-5686){\makebox(0,0)[lb]{\smash{{\SetFigFont{12}{14.4}{\familydefault}{\mddefault}{\updefault}{\color[rgb]{0,0,0}$p$}%
}}}}
\put(5581,-7171){\makebox(0,0)[lb]{\smash{{\SetFigFont{12}{14.4}{\familydefault}{\mddefault}{\updefault}{\color[rgb]{0,0,0}$v_7$}%
}}}}
\put(4231,-7171){\makebox(0,0)[lb]{\smash{{\SetFigFont{12}{14.4}{\familydefault}{\mddefault}{\updefault}{\color[rgb]{0,0,0}$v_1$}%
}}}}
\put(4456,-7171){\makebox(0,0)[lb]{\smash{{\SetFigFont{12}{14.4}{\familydefault}{\mddefault}{\updefault}{\color[rgb]{0,0,0}$v_2$}%
}}}}
\put(4681,-7171){\makebox(0,0)[lb]{\smash{{\SetFigFont{12}{14.4}{\familydefault}{\mddefault}{\updefault}{\color[rgb]{0,0,0}$v_3$}%
}}}}
\put(4906,-7171){\makebox(0,0)[lb]{\smash{{\SetFigFont{12}{14.4}{\familydefault}{\mddefault}{\updefault}{\color[rgb]{0,0,0}$v_4$}%
}}}}
\put(5131,-7171){\makebox(0,0)[lb]{\smash{{\SetFigFont{12}{14.4}{\familydefault}{\mddefault}{\updefault}{\color[rgb]{0,0,0}$v_5$}%
}}}}
\put(5356,-7171){\makebox(0,0)[lb]{\smash{{\SetFigFont{12}{14.4}{\familydefault}{\mddefault}{\updefault}{\color[rgb]{0,0,0}$v_6$}%
}}}}
\put(4343,-6389){\makebox(0,0)[lb]{\smash{{\SetFigFont{12}{14.4}{\familydefault}{\mddefault}{\updefault}{\color[rgb]{0,0,0}$u_1$}%
}}}}
\put(4779,-6389){\makebox(0,0)[lb]{\smash{{\SetFigFont{12}{14.4}{\familydefault}{\mddefault}{\updefault}{\color[rgb]{0,0,0}$u_2$}%
}}}}
\put(5243,-6389){\makebox(0,0)[lb]{\smash{{\SetFigFont{12}{14.4}{\familydefault}{\mddefault}{\updefault}{\color[rgb]{0,0,0}$u_3$}%
}}}}
\put(5683,-6384){\makebox(0,0)[lb]{\smash{{\SetFigFont{12}{14.4}{\familydefault}{\mddefault}{\updefault}{\color[rgb]{0,0,0}$u_4$}%
}}}}
\put(4874,-7592){\makebox(0,0)[lb]{\smash{{\SetFigFont{14}{16.8}{\familydefault}{\mddefault}{\updefault}{\color[rgb]{0,0,0}$(a)$}%
}}}}
\put(6913,-7592){\makebox(0,0)[lb]{\smash{{\SetFigFont{14}{16.8}{\familydefault}{\mddefault}{\updefault}{\color[rgb]{0,0,0}$(b)$}%
}}}}
\put(8947,-7601){\makebox(0,0)[lb]{\smash{{\SetFigFont{14}{16.8}{\familydefault}{\mddefault}{\updefault}{\color[rgb]{0,0,0}$(c)$}%
}}}}
\put(3085,-7048){\makebox(0,0)[lb]{\smash{{\SetFigFont{12}{14.4}{\familydefault}{\mddefault}{\updefault}{\color[rgb]{0,0,0}Level $i+2$}%
}}}}
\put(3088,-6375){\makebox(0,0)[lb]{\smash{{\SetFigFont{12}{14.4}{\familydefault}{\mddefault}{\updefault}{\color[rgb]{0,0,0}Level $i+1$}%
}}}}
\put(3084,-5706){\makebox(0,0)[lb]{\smash{{\SetFigFont{12}{14.4}{\familydefault}{\mddefault}{\updefault}{\color[rgb]{0,0,0}Level $i$}%
}}}}
\put(3084,-5024){\makebox(0,0)[lb]{\smash{{\SetFigFont{12}{14.4}{\familydefault}{\mddefault}{\updefault}{\color[rgb]{0,0,0}Level $i-1$}%
}}}}
\end{picture}%

%% file: LazyLocalTree.pstex_t
\begin{picture}(0,0)%
\includegraphics{LazyLocalTree.pstex}%
\end{picture}%
\setlength{\unitlength}{4144sp}%
\begingroup\makeatletter\ifx\SetFigFontNFSS\undefined%
\gdef\SetFigFontNFSS#1#2#3#4#5{%
  \reset@font\fontsize{#1}{#2pt}%
  \fontfamily{#3}\fontseries{#4}\fontshape{#5}%
  \selectfont}%
\fi\endgroup%
\begin{picture}(5266,1615)(3051,-4548)
\put(6710,-3676){\makebox(0,0)[lb]{\smash{{\SetFigFontNFSS{12}{14.4}{\familydefault}{\mddefault}{\updefault}{\color[rgb]{0,0,0}top tree}%
}}}}
\put(4012,-3105){\makebox(0,0)[lb]{\smash{{\SetFigFontNFSS{12}{14.4}{\familydefault}{\mddefault}{\updefault}{\color[rgb]{0,0,0}$T_h(u)$}%
}}}}
\put(6829,-3114){\makebox(0,0)[lb]{\smash{{\SetFigFontNFSS{12}{14.4}{\familydefault}{\mddefault}{\updefault}{\color[rgb]{0,0,0}$T_l(u)$}%
}}}}
\put(5345,-3180){\makebox(0,0)[lb]{\smash{{\SetFigFontNFSS{12}{14.4}{\familydefault}{\mddefault}{\updefault}{\color[rgb]{0,0,0}$u$}%
}}}}
\end{picture}%

%% file: simplemst.tex
\subsection{Supporting decremental MSF}\label{subsec:DecMST}
%An $\Oo(\log^2n)$ time data structure for MSF: use data structure of Wulff-Nilsen extended such that each node of $\mathcal C_L$ contains the weight of the cheapest level $i$-edge below it, for each $i$; it can be assumed here and below that all weights are distinct and belong to $\{0,1,\ldots,n^2\}$  by doing an initial comparison sort and then working on ranks of weights instead.
We can convert the above fully dynamical connectivity structure to a decremental MSF structure by using a trick from~\cite{holm01connect}. For decremental MSF, we can assume that the initial graph is simple and connected and that all weights are distinct and belong to $\{0,1,\ldots,n^2\}$ by doing an initial comparison sort and then working on ranks of weights instead. All edges start at level $0$ and we initialize the spanning forest to the MSF. When searching through the level $i$ non-tree edges incident to $C(w)$ as in Section~\ref{subsec:InsDel}, we do so in order of increasing weight. We support this by letting each node of $\mathcal C_L$ contain the weight of the cheapest level $i$-edge below it, for each $i$. To find the cheapest non-tree edge with an endpoint in $C(w)$ we can follow the cheapest level $i$ weight down from $w$ in $\mathcal C_L$ until we reach a leaf $x$ and then take the cheapest level $i$-edge incident to $x$. As shown in~\cite{holm01connect}, this small modification to the connectivity structure suffices to support decremental MSF.

\paragraph{Performance}
Finding the initial MSF can be done in $\Oo(m + n\log n)$ time using Prim's algorithm with Fibonacci heaps.
We split the time complexity analysis for the rest of the above data structure into three parts: searching for edges down from $C(w)$ in $\mathcal C_L$ to identify a cheapest level $i$-edge incident to a leaf $x$, maintaining the edge weights associated with nodes of $\mathcal C_L$, and making structural changes to $\mathcal C_L$.

To analyze the time for the first part, note that since $\mathcal C_L$ has height $\Oo(\log n)$, searching down from $C(w)$ to $x$ takes $\Oo(\log n)$ time. In order to efficiently identify the cheapest level $i$-edge incident to such a leaf $x$, we extend the data structure by letting $x$ have an $\Oo(\log n)$ array of doubly-chained lists of edges, so let $E_i(x)$ be the list of level $i$ non-tree edges adjacent to $w$ in order of increasing weight. The cheapest level $i$-edge incident to $x$ is then the first edge of $E_i(x)$ and can thus be found in $\Oo(1)$ time.  When increasing the level of an edge $e=(x,y)$ from $i$ to $i+1$, it is not a replacement edge, and is therefore the cheapest level $i$ non-tree edge adjacent to any vertex in its component.  In particular it is the cheapest level $i$ non-edge incident to $x$ and $y$ and is therefore at the start of $E_i(x)$ and $E_i(y)$. Furthermore, (as shown in~\cite{holm01connect}) it is costlier than all other edges that have been moved to level $i+1$ earlier so when we move it all we need to do is put it at the end of $E_{i+1}(x)$ and $E_{i+1}(y)$ to keep them sorted. This takes $\Oo(1)$ time.

We have shown how the cheapest non-tree edge incident to $C(w)$ can be found in $\Oo(\log n)$ time. Maintaining edge weights associated with nodes of $\mathcal C_L$ can also be done in $\Oo(\log n)$ time since for each edge level change (or the deletion of an edge), only the weights along the leaf-to-root paths in $\mathcal C_L$ from the endpoints of the edge need to be updated. It remains to bound the time for structural changes to $\mathcal C_L$. It was shown in~\cite{Wulff-Nilsen13a} that by picking constant $\epsilon_h$ sufficiently small and constant $\alpha$ sufficiently large (see definitions in Section~\ref{subsec:LocalTrees}), this takes amortized $\Oo(\log n/\log\log n)$ time per edge level change plus an additional $\Oo(\log^2 n/\log\log n)$ worst-case time per edge deletion.

We conclude from the above that the total time to build our decremental MSF structure on a simple connected graph with $n$ vertices and $m$ edges, and then deleting $d$ edges is $\Oo(m\log^2 n + d\log^2 n)$.  In the next section, we give a variant of this data structure where exactly the same structural changes occur in $\mathcal C_L$ but where the time to search for edges is sped up using a new shortcutting system together with fast priority queues. Since structural changes take a total of $\Oo(m\log^2n/\log\log n + d\log^2n/\log\log n)$ time, these will not be the bottleneck so we ignore them in the time analysis in the next section. Also, the structure in~\cite{Wulff-Nilsen13a} can identify the parent cluster node of a cluster in $\Oo(\log n/\log\log n)$ time so we shall also ignore this cost.

%In addition, we let each leaf $x$ of $\mathcal C_L$  We then recompute the cost of the new cheapest level $i$ and $i+1$ edges on each of the paths from $x$ and $y$ to $w$ in $\mathcal C_L$.  Searching for the endpoints of $e$ from $C(w)$ takes $\Oo(\log n)$ timeThe total cost for finding the cheapest non-tree edge in a cluster, and for increasing its level is therefore $\Oo(\log n)$, just like in the connectivity algorithm, and the total time for building the structure with $n$ vertices and $m$ edges, and then deleting $d$ edges is $\Oo(n + m\log^2 n + d\log^2 n)$.

%% file: fastmst.tex
\section{Faster Data Structure for Dynamic MSF}\label{sec:FastMSF}
In this section, we present our new data structure for decremental MSF. Assume that the initial graph is connected. If not, we maintain the data structure separately for each component. The total time bound is $\Oo(m\log^2n/\log\log n + d\log^{2+\epsilon}n)$ for a constant $\epsilon < 1$, where the initial graph has $m$ edges and $n$ vertices and where $d$ edges are deleted in total. By Corollary~\ref{Cor:Reduction}, this suffices in order to achieve $\Oo(\log^4n/\log\log n)$ update time for fully-dynamic MSF. 

%\subsection{Speed-up with fast priority queues}
A bottleneck of the simple data structure for decremental MSF presented in Section~\ref{subsec:DecMST} is moving up and down trees of $\mathcal C_L$. 
The data structure identifies level $i$-edges incident to a level $(i+1)$-cluster $C(u)$ in order of increasing weight by moving down $\mathcal C_L$ from node $u$, always picking the child (or children) with the cheapest level $i$-edge below it. When a leaf is reached, the cheapest level $i$-edge $e$ incident to it is traversed. If both endpoints of $e$ were identified in the downward search then we do not need an upwards search. If only one endpoint was identified then we do an upwards search in $\mathcal C_L$ from the other endpoint until reaching the node for a level $(i+1)$-cluster. Each upwards search can trivially be done in $\Oo(\log n)$ time as this is a bound on the height of trees in $\mathcal C_L$. We claim that this is actually fast enough. To see why, note that we only do an upwards search when a reconnecting edge is found. At most one reconnecting edge is found per edge deleted so we can in fact afford to spend $\Oo(\log^{2+\epsilon}n)$ time on the upwards search. In the following, we can thus restrict our attention to speeding up downward searches. It suffices to get a search time of $\Oo(\log n/\log\log n)$ since for every two downward searches, we either increase the level of an edge or we find a reconnecting edge.

\subsection{A downwards shortcutting system}
We use a downwards shortcutting system with fast min priority queues to speed up downward searches. Certain nodes of $\mathcal C_L$ are augmented with min priority queues keyed on edge weights. Since we may assume that edge weights are in the range $\{0,1,\ldots,n^2\}$, we can use fast integer priority queues. In the following, we assume constant time for each queue operation. As mentioned in the introduction, a less efficient queue will slow down the performance of our data structure by a factor equal to its operation time.

The nodes of $\mathcal C_L$ with associated priority queues are referred to as \emph{queue nodes}. The following types of nodes are queue nodes ($\epsilon_q$ is a small constant to be chosen later):
\begin{enumerate}
\item cluster nodes whose level is divisible by $i\lceil\epsilon_q\log\log n\rceil$ for an integer $i$,
\item heavy tree nodes $u$ with a parent $v$ in $\mathcal C_L$ such that $\rank(u)\leq i\lceil\epsilon_q\log\log n\rceil < \rank(v)$ for an integer $i$,% picked such that $\lceil\epsilon_q\log\log n\rceil$ is an integer,
\item rank nodes of light trees whose rank is divisible by $\lceil\epsilon_q\log\log n\rceil$,
\item roots and leaves of buffer, bottom, and top trees.
\end{enumerate}

Each queue node $u$ (excluding leaves of $\mathcal C_L$) is associated with an array whose $i$th entry points to a min-queue $Q_i(u)$, for each level $i$. If $u$ is a proper ancestor of a level $i$-node, $Q_i(u)$ is empty. Otherwise, for each nearest descending queue node $v$ of $u$ in $\mathcal C_L$, $Q_i(u)$ contains the node $v$ with associated key $k$ denoting the weight of the cheapest level $i$-edge incident to a leaf of $\mathcal C_L$ below $v$.

\paragraph{Traversing the shortcutting system}
The priority queues associated with queue nodes induce our downwards shortcutting system in $\mathcal C_L$. To see how, consider a level $(i+1)$-cluster $C(u)$. To identify the cheapest level $i$-edge $e$ incident to $C(u)$, assume first that $u$ is a queue node. Then a minimum element in $Q_i(u)$ is a node $v$ below $u$ with $e$ incident to $C(v)$. We refer to $(u,v)$ as a \emph{shortcut}. Whereas our simple data structure would traverse the path from $u$ down to $v$ in $\mathcal C_L$, our new data structure can use the shortcut $(u,v)$ to jump directly from $u$ to $v$ within the time it takes to obtain the minimum element in $Q_i(u)$. At $v$, we identify a minimum element $w$ in $Q_i(v)$ and jump directly to this node along $(v,w)$. This shortcut traversal continues until a leaf of $\mathcal C_L$ is reached, and $e$ is identified as one of the edges incident to this leaf. If both endpoints of $e$ are below $u$ in $\mathcal C_L$, one of the queues contains two distinct minimum elements $v$ and $v'$, corresponding to where the paths down to the endpoints of $e$ branch out. In this case, we search down from both $v$ and $v'$.

Now assume that $u$ is not a queue node. Then all nearest descending queue nodes $v$ of $u$ are visited and for each of them the minimum element in $Q_i(v)$ is identified and its associated key $k_i(v)$. The search procedure described above is then applied to each of the at most two nodes $v$ with minimum key $k_i(v)$.

\paragraph{Performance}
Let us analyze the time for the search procedure just described. The following lemma bounds the time to identify the nearest descending queue nodes.
\begin{lemma}\label{Lem:DescendingQueueNodes}
The set of nearest descending queue nodes of a cluster node can be found in $\Oo(\log^{3\epsilon_q}n)$ time.
\end{lemma}
\begin{proof}
Let $u$ be a cluster node and let $T$ be the subtree of $\mathcal C_L$ rooted at $u$ whose leaves are queue nodes and whose non-leaf nodes are not. It suffices to show that $|T| = O(\log^{3\epsilon_q}n)$. Note that since roots of light trees are queue nodes, all non-leaf nodes of $T$ except possibly $u$ belong to heavy trees. Consider a root-to-leaf path $P$ in $T$. Since ranks go strictly down along rank paths and along root-to-leaf paths in rank trees, at most $\lceil\epsilon_q\log\log n\rceil$ edges of $P$ are contained in rank trees or rank paths. Since levels of cluster nodes go strictly down along $P$, there are at most $\lceil\epsilon_q\log\log n\rceil$ cluster nodes on $P$ that are not queue nodes. A traversal of $P$ through a heavy tree encounters at most two edges not belonging to the rank path or a rank tree. Hence $P$ contains at most $3\lceil\epsilon_q\log\log n\rceil$ edges. Since $\mathcal C_L$ is binary, $|T| = O(\log^{3\epsilon_q}n)$.
\end{proof}
If our initial node $u$ is not a queue node, we can thus in $\Oo(\log^{3\epsilon_q}n)$ time find all nearest descending queue nodes of $u$ and among these obtain the at most two nodes $v$ with smallest key in $Q_i(v)$.

%Suppose first that the initial node $u$ is not a queue node. Consider traversing a path $P$ downwards from $u$ until reaching a queue node. Note that $P$ cannot go through a light tree as the root of such a tree is a queue node (as it is a top tree root). In a heavy tree, ranks are strictly decreasing down along rank paths and along root-to-leaf paths in rank trees. Hence, for every node of a heavy tree, its rank is strictly larger than the ranks of its grandchildren (if any). Since the ranks ordered along $P$ cannot increase, $P$ contains no more than $2\lceil\epsilon_q\log\log n\rceil$ edges of heavy trees. From this and from the observation that $\mathcal C_L$ is binary, it follows that the number of edges traversed and the number of nearest descending queue nodes of $u$ is $\Oo(2^{2\lceil\epsilon_q\log\log n\rceil}) = O(\log^{2\epsilon_q}n)$. Since queue operations take constant time, we can identify the at most two queue nodes $v$ with smallest key value in $Q_i(v)$ in $\Oo(\log^{2\epsilon_q}n)$ time.

Now, assume that the initial node $u$ is a queue node and consider the shortcut path $P$ of queue nodes from $u$ to a leaf that the procedure visits. The number of visited queue nodes of type $1$ is clearly $\Oo(\log n/(\epsilon_q\log\log n))$. Since ranks of nodes along $P$ cannot increase and since the difference in rank between two consecutive rank nodes on $P$ is at least $\lceil\epsilon_q\log\log n\rceil$, the number of queue nodes of type $2$ or $3$ is also $\Oo(\log n/(\epsilon_q\log\log n))$. Finally, since the rank difference between a cluster node $u$ and any leaf in $T_l(u)$ is $\Omega(\frac 1{\epsilon_h}\log\log n)$ (see~\cite{Wulff-Nilsen13a}), $P$ contains only $\Oo(\log n/(\epsilon_h\log\log n))$ queue nodes of type $4$. Given our downwards shortcutting system, the cheapest level $i$-edge incident to $C(u)$ can thus be found in $\Oo(\log^{3\epsilon_q}n + (\frac 1{\epsilon_h} + \frac 1{\epsilon_q})\log n/\log\log n)$ time. Below we show how to maintain this system efficiently under changes to $\mathcal C_L$.

\subsection{Dealing with non-topological changes}
Two types of changes occur in $\mathcal C_L$: topological changes when cluster nodes are merged or split and non-topological changes when an edge increases its level or is removed and information about which edges are the cheapest below a cluster node needs to be updated. We start with the non-topological changes.

Suppose a level $i$-edge $e$ disappears, either because it is deleted or because its level is increased to $i+1$. Then we need to update priority queues of queue nodes accordingly. If $\ell(e)$ increases then the two downward paths identified with our shortcutting system contain all the queue nodes whose level $i$-queues need to be updated. For each endpoint $x$ of $e$, we traverse each of these paths bottom-up. Let $u$ be the current non-leaf node in one of these traversals and let $v$ be its predecessor. Note that the key of $v$ in $Q_i(u)$ equals the weight $w(e)$ of $e$. We increase this key to the key for the minimum element in $Q_i(v)$ (or remove $v$ from $Q_i(u)$ if $Q_i(v)$ is empty). Otherwise we stop as no queue nodes above $u$ need updates. As each queue update takes $\Oo(1)$ time, total time is bounded by the number $\Oo((\frac 1{\epsilon_h} + \frac 1{\epsilon_q})\log n/\log\log n)$ of queue nodes considered.

We also need to update priority queues for level $(i+1)$-edges since $e$ has its level increased to $i+1$. Note that all the queue nodes that need to be updated belong to the two downward paths traversed. Again, we traverse each path bottom-up. Let $u$ be the current non-leaf node in one of the traversals and let $v$ be its predecessor. If $v$ is not present in $Q_{i+1}(u)$, we add it with key $w(e)$. Otherwise, if the key of $v$ in $Q_{i+1}(u)$ is greater than $w(e)$, we decrease it to $w(e)$. In both cases, we then proceed upwards. Otherwise, we stop since no queues above $u$ need updates. Total time to update level $(i+1)$-queues is $\Oo((\frac 1{\epsilon_h} + \frac 1{\epsilon_q})\log n/\log\log n)$.

It remains to consider the case where $e$ disappears because it was deleted. Then we identify all the queue nodes above $e$ that need to be updated by traversing the leaf-to-root paths in $\mathcal C_L$ for the endpoints of $e$. The queue nodes visited have their queue nodes updated as described above. Since $\mathcal C_L$ has height $\Oo(\frac 1{\epsilon_h}\log n)$, total time is $\Oo(\frac 1{\epsilon_h}\log n + (\frac 1{\epsilon_h} + \frac 1{\epsilon_q})\log n/\log\log n)$. This completes the description of how to deal with non-topological changes.

\subsection{Dealing with topological changes}
Now, we describe how to maintain queues under topological changes to $\mathcal C_L$. We will assume that deleting a shortcut is free as it is paid for when the shortcut is formed. In our analysis for bounding the total time to form shortcuts, we shall use the accounting method; during the course of the algorithm, credits will be associated with certain parts of $\mathcal C_L$ and each credit can pay for a constant amount of work. Denote by $s_{\max} = \log^\alpha n$ the maximum number of leaves of a buffer tree. The following invariants are maintained:
\begin{itemize}
\item Each leaf of a heavy tree contains $(2 + \log s_{\max})\log n$ credits (\emph{heavy tree invariant}),
\item Each leaf of a buffer tree contains $(2 + \log s_{\max} - \log s)\log n$ credits where $s$ is the number of leaves in the tree (\emph{buffer tree invariant}).
\item Each leaf of a bottom tree contains $1$ credit (\emph{bottom tree invariant}).
\end{itemize}
\begin{lemma}\label{Lem:BufferTreeCredits}
A buffer tree with $s_1$ leaves contains more credits than a buffer tree with $s_2 < s_1$ leaves.
\end{lemma}
\begin{proof}
The function $f(x) = x(2 + \log s_{\max} - \log x)$ is monotonically increasing on $[1, s_{\max}]$ since $f'(x) = 2 + \log s_{\max} - \log x - 1/\ln 2 > \log s_{\max} - \log x\geq 0$ for all $x\in[1, s_{\max}]$.
\end{proof}
Recall that $\epsilon_h$ was introduced when defining heavy and light children. We observe that initially, all edges of the decremental MSF structure have level $0$ and because of our assumption that the initial graph is connected, $\mathcal C$ consists of a single root $r$ with each vertex of the graph as a child, implying that $\mathcal C_L$ is the single local tree $L(r)$. This local tree contains at most $\log^{\epsilon_h} n$ leaves in the heavy tree and a single buffer tree with at most $s_{\max}$ leaves. Furthermore, there are at most $n$ bottom tree leaves. By Lemma~\ref{Lem:BufferTreeCredits}, the initial amount of credits required is at most $\log^{\epsilon_h}n(2 + \log s_{\max})\log n + 2s_{\max}\log n + n$.% which can be paid for by the $m\log^2n/\log\log n$-term in our total time bound.

%The following lemma is crucial to our analysis below.
%\begin{lemma}\label{Lem:ClusterBound}
%After an edge deletion in the decremental MSF structure, only $\Oo(\log n)$ new %clusters are formed.
%\end{lemma}
%\begin{proof}
%At each level $i$, we merge a number of level $(i+1)$ clusters into one. If a replacement level $i$ edge is not found at this level, the merged cluster gets a new level $i$ parent cluster. Hence, at most $2\log n$ new clusters are formed in total over all levels.
%\end{proof}

\subsubsection{Merging cluster nodes}
The general type of change to $\mathcal C$ during the deletion of a level $i$-tree edge was described in Section~\ref{subsec:InsDel} and the corresponding updates to local trees in $\mathcal C_L$ was described in Section~\ref{subsec:LocalTrees}. The first step is to merge all level $(i+1)$ clusters on the smaller component of the split level $i$-tree (Figure~\ref{fig:SplitMerge}(a) and (b)). We now describe how to update shortcuts accordingly. For now, assume that only two level $(i+1)$ clusters $C(u)$ and $C(v)$ are merged into a new level $(i+1)$ cluster $C(w)$. We later extend this to the merge of an arbitrary number of clusters. It may be helpful to consult Figure~\ref{fig:LazyLocalTree} in the following.

\paragraph{Shortcuts through the heavy tree}
We say that a shortcut $(x,y)$ \emph{goes through} a node $z\in\mathcal C_L$ if $x$ is an ancestor of $z$ and $y$ is a descendant of $z$ (where possibly $x = z$ or $y = z$). In the new local tree $L(w)$, we obtain all shortcuts through nodes of $T_h(w)$ in a bottom-up manner. Note that queue nodes in the subtrees of $\mathcal C_L$ rooted at leaves of $T_h(w)$ need not be updated. For each queue node $a\in T_h(w)$, assume that all queues of its nearest descending queue nodes have been constructed. Then for each level $j$, we construct $Q_j(a)$ in a brute-force manner by visiting all nearest descending queue nodes $b$ of $a$ and for each of them adding the cheapest node of $Q_j(b)$ to $Q_j(a)$. By Lemma~\ref{Lem:DescendingQueueNodes}, this takes $\Oo(\log^{3\epsilon_q}n)$ time for each $j$, giving a total time of $\Oo(\log^{1 + 3\epsilon_q}n)$ to construct the queues associated with $a$. Since $T_h(w)$ has size $\Oo(\log^{\epsilon_h}n)$, total time to construct all shortcuts through nodes of $T_h(w)$ is $\Oo(\log^{1 + 3\epsilon_q + \epsilon_h}n)$ which over all levels is $\Oo(\log^{2 + 3\epsilon_q + \epsilon_h}n)$. Adding $(2 + \log s_{\max})\log n$ credits to each leaf of $T_h(w)$ is dominated by the cost to construct shortcuts.% By Lemma~\ref{Lem:ClusterBound}, the deletion of $e$ can pay for this cost since it will cost $\Oo(\log^{1 + \epsilon_h}n\log\log n)$ per level which is $\Oo(\log^{2 + \epsilon_h}n\log\log n)$ over all new clusters created after the deletion of $e$.

\paragraph{Shortcuts through the light tree}
Next we describe how to form shortcuts through $T_l(w)$. Let $B_u$ resp.~$B_v$ be the buffer trees of $u$ and $v$, respectively, before the merge. The leaves of the buffer tree $B_w$ of $w$ is the union of leaves of $B_u$ and $B_v$ as well as possibly some leaves from $T_h(u)$ and $T_h(v)$. For now, assume that we obtain $B_w$ simply as the union of $B_u$ and $B_v$, and that $B_u$ and $B_v$ together has at most $s_{\max}$ leaves. Let $r_u$ resp.~$r_w$ be the roots of $B_u$ resp.~$B_w$. Assume w.l.o.g.~that the number $s_u$ of leaves of $B_u$ is smaller than the number of leaves of $B_v$. Tree $B_w$ is formed by adding each leaf of $B_u$ to $B_v$ one by one. As each leaf $l$ is added to $B_v$, we also add shortcuts of the form $(r_u,l)$ to $B_v$. Total time to add all shortcuts is $\Oo(s_u\log n)$. To see that we can afford this, observe that the leaves from $B_v$ will not require more credits when added to $B_w$ since $B_w$ contains at least as many leaves as $B_v$. Before the merge, $B_u$ has $c_u = (2 + \log s_{\max} - \log s_u)\log n$ credits per leaf. Since $B_w$ has at least $2s_u$ leaves, each of its leaves requires at most $(2 + \log s_{\max} - \log (2s_u))\log n = c_u - \log n$ credits so spending $\log n$ credits per leaf of $B_u$ pays for the $\Oo(s_u\log n)$ time spent on the merge.

If $B_u$ and $B_v$ together have more than $s_{\max}$ leaves, we do the same but the result is a new bottom tree. By definition, $B_u$ and $B_v$ each contain at most $s_{\max}$ leaves before the merge so there is at least $1$ credit left on each leaf of $B_w$ after the merge. Hence, the bottom tree invariant is satisfied for the new bottom tree.

Above we assumed that $B_w$ was simply the union of $B_u$ and $B_v$. Now assume that in addition it contains leaves from $T_h(u)$ and $T_h(v)$. By our heavy tree invariant, each such leaf has $(2 + \log s_{\max})\log n$ credits. After having merged $B_u$ and $B_v$, consider adding these leaves one by one, regarding each of them as a trivial buffer tree with $s = 1$ leaf. Then it has the amount of credits required by the buffer tree invariant so the same analysis as above shows that the credits on each leaf can pay for all the required updates.

The remaining shortcuts through $L(w)$ that we need to form are those incident to a top tree node or to a rank node in $T_l(w)$. We use the same analysis as by Thorup~\cite{thorup00connect} and Wulff-Nilsen~\cite{Wulff-Nilsen13a}: note that a bottom tree $B$ does not cause changes to the rank tree above it unless the rank of its root changes. Since this rank is the maximum rank of leaves in $B$, and since this maximum can only decrease (we never add new leaves to $B$ and ranks of existing leaves cannot increase), $B$ causes at most $\log n$ changes to the at most $\log n$ rank nodes above it. By our bottom tree invariant, each bottom tree initially has at least $\log^\alpha n$ credits. Since we are free to pick constant $\alpha$ as large as we like, we may assume that each rank node update in $T_l(w)$ (including top tree leaves excluding the root of the buffer tree) can be paid for by $\log^{\beta} n$ credits where $\beta$ is a constant (growing with $\alpha$) that can be picked as large as we like. Recall that a rank node $a$ in $T_l(w)$ is a queue node if its rank is divisible by $\lceil\epsilon_q\log\log n\rceil$ so the number of its nearest descending queue nodes is at most $2\log^{\epsilon_q} n$. Using a similar analysis as for the heavy tree, we can build all queues for $a$ in $\Oo(\log^{1 + \epsilon_q}n)$ time which can be paid for by the $\log^{\beta}n$ credits on $a$, for sufficiently big $\beta$.

The only new shortcuts not accounted for above are those ending in the root of the top tree and the root of the buffer tree of $T_l(w)$. There are at most $2\log n$ of these and they can be formed in $\Oo(\log n)$ time.% which we can afford by Lemma~\ref{Lem:ClusterBound}.

Above, we have shown how to update our shortcutting system within the desired time bound when exactly two clusters are merged. It is straightforward to extend this to an arbitrary number of clusters. To see this, note that all clusters are merged into a single cluster $w$ so we only get a single heavy tree $T_h(w)$ in the end. Updating shortcuts through $T_h(w)$ and assigning credits to its leaves is done as above. We need to merge multiple buffer trees as well as adding leaves from heavy trees as leaves of buffer trees. Exactly the same analysis as above carries through if we simply view these merges pairwise.

\paragraph{Shortcuts for the parent cluster}
Having updated all the shortcuts through $L(w)$, what remains is to update shortcuts through $L(p)$ where $p$ is the parent of $w$ in $\mathcal C$. The topological changes occuring in $L(p)$ consist of deletions of leaves corresponding to the clusters merged into $w$, and the addition of a single new leaf, namely $w$. Deleting a leaf from a bottom tree is free since it only requires deleting shortcuts from the root of the bottom tree to the deleted leaf and shortcut deletions have zero cost. Deleting leaves from $T_h(p)$ may cause topological changes to this tree but only shortcuts through $T_h(p)$ are affected and these are found as above. Finally, consider deletions of leaves from the buffer tree $B_p$ of $T_l(p)$. By Lemma~\ref{Lem:BufferTreeCredits}, the amount of credits in $B_p$ cannot increase by these deletions and since deletions of shortcuts is free, updating shortcuts through $B_p$ is free as well. Finally, computing shortcuts through $w$ takes $\Oo(\log n)$ time as argued above.

\subsubsection{Splitting cluster nodes}
Above we have shown how to efficiently maintain the shortcutting system when a set of level $(i+1)$ clusters are merged into one cluster, $C(w)$. If a replacement level $i$-edge was found, no more topological changes happen to $\mathcal C_L$ so assume such an edge was not found. Then $w$ needs to be removed as a child of its parent cluster $p$ and added as a child of a new cluster node $p'$ which becomes the sibling of $p$ (Figure~\ref{fig:SplitMerge}(b) and (c)). The removal of $w$ decreases $n(p)$ which may result in some leaves of $T_l(p)$ becoming leaves of $T_h(p)$. Note that this happens for at most $\log^{\epsilon_h}n$ leaves so we can pay for updating all shortcuts through $L(p)$ using the same arguments as above. The only remaining shortcuts that need to be found are those through $L(p')$. As this tree contains only a single leaf, namely $w$, these shortcuts can be found in $\Oo(\log n)$ time. Finally, since $n(p)$ is decreased, it may need to be moved from $T_h(p'')$ to $T_l(p'')$ where $p''$ is its parent in $\mathcal C$. Updating shortcuts accordingly does not increase the overall time bound.

\subsubsection{Performance}
At initialization, we pay $\Oo(m + n\log n)$ for finding the MSF and $\Oo(\log^{\epsilon_h}n(2 + \log s_{\max})\log n + 2s_{\max}\log n + n)$ for the initial amount of credits. The latter is $\Oo(n)$ since $\epsilon_h$ is constant.
For each edge level increase, we pay $\Oo(\log^{3\epsilon_q}n + (\frac 1{\epsilon_h} + \frac 1{\epsilon_q})\log n/\log\log n)$ for searching for the edge and for making non-topological changes to $\mathcal C_L$. Hence an edge pays a total of $\Oo(\log^{1 + 3\epsilon_q}n + (\frac 1{\epsilon_h} + \frac 1{\epsilon_q})\log^2n/\log\log n)$ for this over all its level increases.
For each edge deletion, we pay $\Oo(\log^{1+3\epsilon_q + \epsilon_h}n)$ per level for topological changes to $\mathcal C_L$, giving a total cost of $\Oo(\log^{2+3\epsilon_q + \epsilon_h}n)$ over all levels.
It follows from these calculations that if $d$ edges are deleted, total cost is $\Oo(m(\log^{1 + 3\epsilon_q}n + (\frac 1{\epsilon_h} + \frac 1{\epsilon_q})\log^2n/\log\log n) + d\log^{2+3\epsilon_q + \epsilon_h}n)$. Picking constants $\epsilon_h < \frac 1 2$ and $\epsilon_q < \frac 1 6$, our main result follows from Corollary~\ref{Cor:Reduction}.
\begin{theorem}\label{Thm:Main}
There is a data structure for fully-dynamic minimum spanning tree which supports updates in $\Oo(\log^4n/\log\log n)$ amortized time, assuming the RAM model with standard instructions.
\end{theorem}